\documentclass[11pt]{article}
\usepackage[margin=1in]{geometry}

\usepackage{amsmath}
\usepackage{amsfonts}
\usepackage{amsthm}
\usepackage{amssymb}
\usepackage{url}
\usepackage{hyperref}
\usepackage{color}
\usepackage{graphicx}
\usepackage{xfrac}

\usepackage{algorithm}
\usepackage{algpseudocode}
\usepackage{wrapfig}

\usepackage{tkz-graph}

\input{Qcircuit}

\newcommand{\cclass}[1]{\mathsf{#1}}

\usepackage{color}

\newtheorem{definition}{Definition}

\newtheorem{theorem}{Theorem}

\newtheorem{lemma}{Lemma}

\title{Interactive proofs for BQP via self-tested graph states}
\author{
Matthew McKague \\ 
Jack Dodd Centre for Quantum Technology \\ 
Department of Physics \\
University of Otago \\\
and \\
Centre for Quantum Technologies \\
National University of Singapore \\
\url{matthew.mckague@otago.ac.nz}}
\date{}
\begin{document}
\maketitle
\begin{abstract}
Using the measurement-based quantum computation model, we construct interactive proofs with non-communicating quantum provers and a classical verifier.  Our construction gives interactive proofs for all languages in $\cclass{BQP}$ with a polynomial number of quantum provers, each of which, in the honest case, performs only a \emph{single measurement.}  Our techniques use self-tested graph states which allow us to test the provers for honesty, establishing that they hold onto a particular graph state and measure it in specified bases.  In this extended abstract we give an overview of the construction and proofs.
\end{abstract}

\section{Introduction}
We seek interactive proofs between quantum provers and classical verifiers, both limited to polynomial time calculations.  That is to say, we would like to have a procedure where a classical computer (the ``verifier''), limited to a polynomial number of operations, can query a quantum computer (the ``prover''), also limited to a polynomial number of operations, and tap into its resources in order to perform some computation.  Additionally, if the verifier interacts with a malicious quantum computer it should be able to detect this and abort the calculation.

This problem is interesting for a variety of reasons.  First, as a complexity theoretic question it has obvious value in further developing the theory of how powerful quantum computers are.  From a more practical computing point of view, it would be nice to know whether it would be possible to have cheap classical computers interact with large, and presumably more expensive, quantum ``servers,'' paying for services as required.  Of course the users would like to know that they get their money's worth, and interactive proofs can confirm this.

Another important problem that interactive proofs can address is how to deal with the increasing complexity of quantum experiments.  Current techniques require an exponential amount of classical resources to either compute prediction for experiments, or to characterize quantum devices using tomography.  The crux of the problem is this: we are interested in building quantum computers because we think that they will be exponentially faster than classical computers for some problems, but this very fact will make it impossible for us to classically compute predictions for complex quantum experiments.  Can we still say that quantum mechanics is a falsifiable theory if this were the case  \cite{Aharonov:2012:Is-Quantum-Mech}?  Interactive proofs can give us a way out of this problem: if we can build a quantum computer and design suitable interactive proofs, then we do not need to trust the operation of the quantum prover to accept its results - the interactive proof certifies that the outcome is correct.  We can then safely use the quantum computer to generate predictions for us, for example by simulating the system under study.

For these applications we would like to have interactive proofs for every language that quantum computers can recognize, i.e.\ every language in $\cclass{BQP}$. Clearly the set of languages recognizable by a poly-time classical verifier and poly-time quantum prover lies somewhere between $\cclass{P}$ and $\cclass{BQP}$ since on one hand the verifier can ignore the prover, and on the other hand the verifier and honest prover together form a poly-time quantum machine.  As well, there do exist interactive proofs for all of $\cclass{BQP}$ since $\cclass{BQP} \subseteq \cclass{PSPACE}$ \cite{Bernstein:1997:Quantum-Complex} and $\cclass{PSPACE} = \cclass{IP}$ \cite{Shamir:1992:IP--PSPACE,Lund:1990:Algebraic-metho}, but the known constructions require the prover to solve $\cclass{PSPACE}$-complete problems.  Constructions for particular problems are known (\cite{McKague:2010:Interactive-pro} for example) and of course anything in $\cclass{NP}$ has a trivial interactive proof where the prover simply sends the verifier a certificate.

Current techniques 
\cite{
	Acin:2007:Device-Independ,
	Bardyn:2009:Device-independ,
	Magniez:2006:Self-testing-of,
	Mayers:2004:Self-testing-qu,
	Mayers:1998:Quantum-Cryptog,
	McKague:2012:Robust-self-tes,
	McKague:2010:Self-testing-gr,
	Miller:2012:Optimalrobustquantum,
	Pironio:2009:Device-independ,
	Pironio:2010:Random-numbers-,
	Reichardt:2012:A-classical-lea}
for probing the behaviour of adversarial quantum systems all rely on entanglement, and hence in order to make use of them we must use at least two non-communicating provers to verify entanglement.  Reichardt et al.\ \cite{Reichardt:2012:A-classical-lea} considered the case of two provers.  Here we will consider the case of a polynomial number of provers, but each limited to a \emph{single operation}, and show that we can recognize all languages in $\cclass{BQP}$ with this model.

In this extended abstract we give an overview of the constructions and proofs.  For additional details and a full error analysis we refer the reader to \cite{McKague:2013:Interactive-proofs-for-BQP-via-self-tested-graph-states}.  

\vspace{0.2cm}
\noindent
\textbf{Previous work - } 
Self-testing is a concept for testing quantum apparatus solely through classical interaction with minimal assumptions, such as non-signalling.  It was introduced by Mayers and Yao \cite{Mayers:1998:Quantum-Cryptog, Mayers:2004:Self-testing-qu} who gave a self-test for maximally entangled pairs of qubits and certain measurements. Meanwhile, van Dam et al.\ \cite{van-Dam:1999:Self-Testing-of} considered testing gates in the context of known Hilbert space dimension.  Magniez et al.\ \cite{Magniez:2006:Self-testing-of} combined the two approaches, allowing testing of entire quantum circuits.  Further development of these techniques appear in \cite{Bardyn:2009:Device-independ,McKague:2010:Quantum-Informa, McKague:2011:Generalized-Sel,McKague:2012:Robust-self-tes,Miller:2012:Optimalrobustquantum}.  Self-testing of graph states, critical for our application, appears in \cite{McKague:2010:Self-testing-gr}.

These works all require additional assumptions.  In particular, they assume that devices can be used repeatedly in an independent and identical manner in order to gather necessary statistics.  As well, \cite{Magniez:2006:Self-testing-of} assumes that certain states are in a product form.  

Stemming from a different heritage, Broadbent et al.\ \cite{Broadbent:2008:Universal-blind} considered a semi-quantum verifier who only prepares single qubit states, and a fully quantum prover.  They give a construction for an interactive proof for any language in $\cclass{BQP}$.  Additionally, they describe (without rigorous proof) an extension using two quantum provers and a classical verifier.  Aharonov et al.\ \cite{Aharonov:2008:Interactive-Pro} also describe a semi-quantum protocol using a constant sized quantum verifier and a polynomial-time quantum prover.

Most recently, Reichardt et al.\ \cite{Reichardt:2012:A-classical-lea} proved a very general result showing that two non-communicating quantum provers along with a classical verifier can recognize all languages in $\cclass{BQP}$.  The core of their result is a self-test, using only two provers, for multiple EPR pairs and measurements. Using this tool they show how to test individual gates and perform unitaries via teleportation.  Finally, they combine the results to give an interactive proof for entire quantum circuits.

\vspace{0.2cm}
\noindent
\textbf{Contributions - } 
Our main contribution is to prove the following theorem:

\begin{theorem}
\label{theorem:bqpinteractiveproofs}
For every language $L \in \cclass{BQP}$ and input $x$ there exists a poly$(|x|)$-time verifier $V$ which interacts with a poly$(|x|)$ number of non-communicating quantum provers such that
\begin{itemize}
	\item If $x \in L$ then there exists\footnote{
			The honest provers and the verifier are, of course, members of a uniform set, i.e.\ a description of the verifier and provers can be generated by a polynomial-time Turing machine.
		} 
		a set of honest quantum provers, each of which performs a single one-qubit measurement, for which $V$ accepts with probability at least $c=\sfrac{2}{3}$.
	\item If $x \notin L$ then, for any set of quantum provers, $V$ accepts with probability no more than $s=\sfrac{1}{3}$.
\end{itemize}
\end{theorem}

The proof of this result uses two main pieces.  Using self-testing, we verify that the provers are using low-noise resources sufficient for our computation (i.e.\ graphs states and certain single-qubit measurements).  We then combine this with measurement-based quantum computation in order to construct an interactive proof.

The bulk of our contribution lies in improvements to self-testing.  We modify the proof for the graph state self-test from \cite{McKague:2010:Self-testing-gr}, allowing a tighter error analysis.  For graphs on $n$ vertices the error in the state is upper bounded by $O(\sqrt{n})\epsilon^{\frac{1}{4}}$ (where $\epsilon$ bounds the noise in the experimental outcomes) rather than $O(2^{\frac{n}{2}}) \epsilon^{\frac{1}{2}}$ as in \cite{McKague:2010:Self-testing-gr}.  This exponential improvement in the error scaling in $n$ makes it possible to self-test with a polynomial number of trials to achieve a constant error.  We also analyse the error in the case of adaptive measurements, which are required for measurement-based quantum computing.  Additionally we extend the graph state test to $X$-$Z$ plane measurements  in order to achieve universal computation.

Compared with the result of Reichardt et al.\ \cite{Reichardt:2012:A-classical-lea} our contribution is to provide a different construction with different underlying computational model, that of measurement-based quantum computation.  While they use a constant number of provers, each of which runs in polynomial time, we use a polynomial number of provers, each of which runs in constant time.  Indeed, each prover only performs a single one-qubit measurement.  Our construction also has the advantage that, since only measurements are used, there is no need for additional tests for gates.  As well, the provers are very easy to implement, requiring only the ability to measure in four different bases (once an appropriate graph state is prepared).  

Finally, our construction has a very nice property, which is that \emph{the measurement-based calculation that is performed is exactly what would be done with trusted devices}, whereas the Reichardt et al.\ construction requires qubits to be teleported between the two provers at each gate. This means that an experimentalist can build a quantum computer according to the trusted model, and interactive proofs can be added solely by modifying the classical control\footnote{This is assuming that the experimentalist provides all the required measurement bases, $\{X,Z, \frac{X\pm Z}{\sqrt{2}}\}$, which are required for universal computation anyway.}. 

\section{Construction of the interactive proof protocol}
Let us begin with a language $L \in \cclass{BQP}$ and an input $x$.  Then there exists some procedure CALCULATE$(L, x)$ that the quantum provers can perform which determines whether $x \in L$ and is wrong at most $\sfrac{1}{3}$ of the time.  However since we do not trust the provers, we have to check somehow whether they are honest.
\begin{definition}
A set of provers are said to be \emph{honest} if they follow the procedure CALCULATE$(L, x)$ exactly. If $x \notin L$ then a set of provers is called \emph{dishonest} if the probability of the provers accepting on CALCULATE$(L, x)$ is more than $\sfrac{1}{2}$.  \end{definition}
Note that this definition is asymmetric since dishonest provers only come into play when we consider soundness (the probability of a false positive.)  Also, there can be provers which are neither honest nor dishonest, which we can think of noisy honest provers.

Now suppose we have a test for honesty which rejects dishonest provers.  If we apply this test and find that the provers are indeed honest, we can ask them to run the calculation and accept whenever they say to accept.  If the provers are dishonest then we immediately reject.  Of course we must be more careful then this since dishonest provers could pretend to be honest up until we perform the calculation, and then give an incorrect answer.  To get around this, we structure our test for honesty so that an individual prover cannot distinguish the calculation from the test\footnote{To be more specific, an individual prover should not have any procedure, better than just guessing, which positively identifies that the calculation is taking place.   It is, however, acceptable if an individual prover can tell that we are performing TEST$(L,x)$ since this information only motivates it to perform honestly!}.  We then randomly choose whether to test or calculate.  The overall protocol is given in algorithm~\ref{algorithm:interactiveproof}.

\begin{algorithm}
\caption{INTERACTIVEPROOF$(x)$}
\label{algorithm:interactiveproof}

\begin{algorithmic}
	\State Randomly choose CALCULATE$(L, x)$ with probability $q$ or TEST with probability $1-q$
	\If {CALCULATE} 
		\State
		\Return CALCULATE$(L, x)$
	\Else
		\State
		\Return TEST$(L, x)$
	\EndIf
\end{algorithmic}
\end{algorithm}

In order for our procedure to work, we need to verify some properties:
\begin{itemize}
	\item The procedure CALCULATE$(L, x)$ must be in a form we can test.			
	\item The procedure TEST$(L, x)$ must accept honest provers with a higher probability than it accepts dishonest provers, and the gap between these two probabilities must be at least $\sfrac{1}{\text{poly}(|x|)}$ so that we can distinguish between them with a polynomial number of trials.
	\item An individual prover must not be able to distinguish, based on its inputs, when we are performing procedure CALCULATE$(L,x)$.
\end{itemize}
Assuming these two conditions hold, we can prove the following lemma.

\begin{lemma}
\label{lemma:interactiveproofoneshot}  Let $L$ be a language and $x$ in input.  Then there exists some $1 \geq q \geq 0$ 
and some $1 \geq c_{ip} > s_{ip} \geq 0$ such that
\begin{itemize}
	\item If $x \in L$ then for honest provers procedure~\ref{algorithm:interactiveproof} accepts with probability at least $c_{ip}$
	\item If $x \notin L$ then for any set of provers procedure~\ref{algorithm:interactiveproof} accepts with probability at most $s_{ip}$
\end{itemize}
and
\begin{equation}
	c_{ip} - s_{ip} \geq 
	\frac{1
	}{
		poly(|x|)
	}
\end{equation}
\end{lemma}

The proof is a straightforward case analysis and optimization over $q$.  Roughly, the two cases are honest provers versus noisy honest provers, which favours a high $q$ emphasizing the computation, and honest provers versus dishonest provers, which favour a low $q$ emphasizing the test for honesty.

Our gap grows vanishingly small as $|x|$ increases, but we can amplify to a constant gap by repeating the procedure $\text{poly}(|x|)$ times as in Algorithm~\ref{algorithm:gapamplification}.

\begin{algorithm}
\caption{AMPLIFYGAP$(L, x)$}
\label{algorithm:gapamplification}

\begin{algorithmic}
	\State Perform Algorithm~\ref{algorithm:interactiveproof} $N = \text{poly}(|x|)$ times and let the number of times the algorithm accepts be $M$.
	\If {$M > N\frac{c_{ip} - s_{ip}}{2}$} 
		\State
		\Return ACCEPT
	\Else
		\State
		\Return REJECT
	\EndIf
\end{algorithmic}
\end{algorithm}

\begin{lemma}
If $x \in L$ and the provers are honest then Algorithm~\ref{algorithm:gapamplification} accepts with probability at least $\sfrac{2}{3}$ and if $x \notin L$ then it accepts with probability at most $\sfrac{1}{3}$.
\end{lemma}

This is a standard result, which can be shown via a straightforward application of Hoeffding's inequality.  The repetition can be done either serially, or in parallel by adding a new set of provers for each repetition.  In the latter case each prover is queried at most once.

Our remaining tasks are to show that a suitable calculation and test for honesty exist.

\section{The Calculation}
We do not make any contributions in the section.  Instead we give sufficient details of the computational model that we are using for the reader to understand the test for honesty.

Rather than using the standard circuit model of quantum computation, we will be using the measurement-based quantum computation model, introduced by Raussendorf et al.\ \cite{Raussendorf:2001:A-One-Way-Quant, Raussendorf:2003:Measurementbasedquantum}.  The basic idea is to use teleportation to apply unitaries.  By varying the measurements used in the teleportation one can induce various unitaries on the qubit being teleported.  Hence gates are replaced by measurements.  Of course teleportation requires correction gates, but these can be absorbed into the later behaviour of the circuit.  By carefully accounting for correction gates, one can perform an entire calculation using only measurements on a suitably prepared resource state.  This resource state depends only on the size of the circuit one wishes to perform and are examples of \emph{graph states}. 

\vspace{0.2cm}
\noindent
\textbf{Graph states - } 
We assume that the reader is familiar with the basics of graph theory.  A good resource is \cite{Diestel:2010:Graph-Theory}.  We now fix some notation.  Let $G = (V, E)$ be a graph, $n = |V|$ and $u, v \in V$. 
%
%
%
The \emph{adjacency matrix} $\mathbf{A}$ of $G$ is a $\{0,1\}$ matrix with $\mathbf{A}_{u,v} = 1$ whenever $(u,v) \in E$ and 0 elsewhere.  Setting $1_v$ to be the vector with 1 in position $v$ and zero elsewhere, the vector $\mathbf{A}1_{v}$ contains a 1 in position $u$ for each $(u,v) \in E$, and is hence the characteristic vector of the \emph{neighbourhood} of $v$, the set of vertices adjacent to $v$.  A \emph{triangle} is a set of three vertices, all of which are adjacent.

The graph state $\ket{G}$ on a graph $G$ is the $n$-qubit state (each qubit is associated with a vertex of $G$) given by
\begin{equation}
	\ket{G} = 
	\frac{1}{\sqrt{2^{n}}}
	\sum_{x \in\{0,1\}^{n}} 
		(-1)^{\frac{1}{2}x \cdot \mathbf{A}x}\ket{x}.
\end{equation}

We can equivalently specify $\ket{G}$ using the \emph{stabilizer formalism} \cite{Gottesman:1997:Stabilizer-Codes-and-Quantum-Error-Correction}.  Setting $X$, $Y$ and $Z$ to be the Pauli matrices
\begin{equation}
	X = 
	\left(
		\begin{matrix}
		0 & 1 \\
		1 & 0 \\
		\end{matrix}
	\right)
	,
	\quad
		Y = 
	\left(
		\begin{matrix}
		0 & -i \\
		i & 0 \\
		\end{matrix}
	\right)
	,
	\quad
	Z = 
	\left(
		\begin{matrix}
		1 & 0 \\
		0 & -1 \\
		\end{matrix}
	\right)
\end{equation}
the \emph{stabilizers} of $\ket{G}$ are the $n$-fold tensor product Pauli operators $P$ such that $P \ket{G} = \ket{G}$.  The stabilizers form an Abelian group, called the \emph{stabilizer group}, which has $2^n$ elements and is generated by a set of $n$ elements.  By specifying the stabilizer group we fix a 1-dimensional subspace and hence essentially fix a particular state.

Let us use the shorthand notation $M^{x} = \bigotimes_{j=1}^{n} M_{j}^{x_{j}}$, where $x$ is an $n$-bit string, and $M_j$ operates on the $j$-th subsystem.  Then the stabilizer group for $\ket{G}$ is generated by the operators
\begin{equation}
	S_{v} = X_{v} Z^{\mathbf{A}1_{v}}.
\end{equation}
That is, $S_{v}$ has $X$ on vertex $v$ and $Z$ on each of the neighbours of $v$.  The stabilizers are products of the generators, and have the form $(-1)^{\frac{1}{2} t \cdot \mathbf{A} t}X^tZ^{\mathbf{A}t}$ for some $n$-bit string $t$.  In particular, if $T$ is a triangle in $G$ with characteristic vector $\tau$ then the operator $-X^\tau Z^{\mathbf{A} \tau}$ is a stabilizer for $\ket{G}$.	

The graph states that we will be considering are \emph{triangular cluster states} which are graph states where the underlying graph is a triangular lattice such as in figure~\ref{figure:triangularlattice}.

\begin{figure}
\begin{center}
\begin{tikzpicture}[node distance=1cm]
	\GraphInit[vstyle=Simple]
	\Vertices[x=0, y=0]{line}{a0,b0,c0,d0,e0}
	\Vertices[x=0.5, y=1]{line}{a1,b1,c1,d1,e1}
	\Vertices[x=0, y=2]{line}{a2,b2,c2,d2,e2}
	\Vertices[x=0.5, y=3]{line}{a3,b3,c3,d3,e3}
	\Edges(a0,b0,c0,d0,e0)
	\Edges(a1,b1,c1,d1,e1)
	\Edges(a2,b2,c2,d2,e2)
	\Edges(a3,b3,c3,d3,e3)

	\Edges(a0,a1,a2,a3)
	\Edges(b0,b1,b2,b3)
	\Edges(c0,c1,c2,c3)
	\Edges(d0,d1,d2,d3)
	\Edges(e0,e1,e2,e3)
	
	\Edges(b0,a1,b2,a3)
	\Edges(c0,b1,c2,b3)
	\Edges(d0,c1,d2,c3)
	\Edges(e0,d1,e2,d3)
\end{tikzpicture}
\end{center}
\caption{A triangular lattice graph}
\label{figure:triangularlattice}
\end{figure}
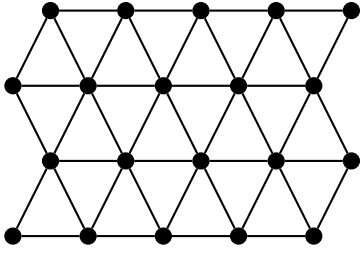

\vspace{0.2cm}
\noindent
\textbf{Triangular cluster state model - }
We will make use of the following theorem, due to Mahalla and Perdrix \cite{Mhalla:2012:Graph-States-Pi}, which allows us to transform any quantum circuit into a measurement-based quantum computation in a form which we can test.

\begin{theorem}[Mahalla and Perdrix \cite{Mhalla:2012:Graph-States-Pi}] \label{theorem:graphstatecomputation}
Triangular cluster states are universal resources for measurement-based computation based on measurements $X, Z, \frac{X\pm Z}{\sqrt{2}}$, and the number of vertices in the cluster state is polynomial in the size of the original circuit we wish to perform.
\end{theorem}

The algorithm for making a measurement-based computation in this model is given in Algorithm~\ref{algorithm:calculate}. The measurement settings for each qubit $M(L, x, a_1, \dots a_{v-1}) \in \{X, Z, \frac{X\pm Z}{\sqrt{2}}\}$ and the output RESULT$(L, x, a_1 \dots a_n)$ can be computed in polynomial time.  An important aspect of this algorithm is that the measurement settings are adaptive, meaning that the measurement setting for vertex $v$ can depend on the outcome of previously applied measurements.

\begin{algorithm}
\caption{CALCULATE$(L, x)$}
\label{algorithm:calculate}

\begin{algorithmic}
	\State Prepare a triangular cluster state of size $n = \text{poly}(|x|)$
	\For {$v=1$ to $n$} 
		\State Measure vertex $v$ in basis $M(L, x, a_{1} \dots a_{v-1})$ to obtain outcome $a_{v}$
	\EndFor
	\State Output RESULT$(L, x, a_1 \dots a_n)$.
\end{algorithmic}
\end{algorithm}

In order to put CALCULATE$(L,x)$ into a form that we can test, we distribute the $n$-qubit cluster state to $n$ provers, one vertex per prover.  Honest provers will accept a measurement setting in $\{X, Z, \frac{X\pm Z}{\sqrt{2}}\}$, measure their qubit in the specified basis, and return the result as $\pm 1$.  For the dishonest case, we will assume that the $n$ provers are non-communicating and are limited to quantum operations (no ``post-quantum'' operations allowed.)  We put no further restrictions on the provers, and in particular we make no assumptions about the dimension of the Hilbert space of their states.

\section{Test for honesty}
\label{sec:selftesting}

First, we introduce a bit of notation.   Honest provers will each hold one qubit corresponding to a vertex of the cluster state $\ket{G}$.  The honest prover for $v$ will measure its qubit according to observable $M$.  We will abuse notation slightly and also consider ``$M$'' to be the symbol that we use to query the prover when we want it to measure observable $M$.

For untrusted provers, the $n$-partite state $\ket{\psi^\prime}$ is the state that is actually held by the $n$ provers.  It is sufficient to consider pure states since we make no assumptions about the dimension of the Hilbert spaces, so we can give the purification of any mixed state to one of the provers.  The operator $M^\prime_v$ is the observable that the prover for vertex $v$ uses to measure its state when queried with measurement setting $M$.  The provers respond with $\pm 1$, and so we can model their actions as a two-outcome observable $M^\prime_v$ with eigenvalues $\pm 1$.  Hence $(M^\prime_v)^2 = I$ and $M^\prime_v$ is unitary.

The structure of the test for honesty is very straightforward and consists of choosing one of a set of measurement settings at random and accepting if the result matches what we expect from the honest case.  To determine the measurement settings that are part of the test we first determine the triangular lattice graph $G$ that is needed, which depends only on the size of the circuit we wish to perform.  Next, fix a set $T$ of triangles in $G$ so that each vertex of $G$ lies it at least one triangle.  Finally, for each vertex $v$ fix a neighbour $u$.  Then the set of measurement settings $\mathcal{M}(L, x)$ consists of:
\begin{itemize}
	\item the stabilizer generators $S_v = X_v Z^{\mathbf{A} 1_v}$ for the graph state $\ket{G}$
	\item stabilizers $-X^{\tau} Z^{\mathbf{A} \tau}$ where $\tau$ is the characteristic vector for a triangle in $T$
	\item measurements $\left(\frac{X \pm Z}{\sqrt{2}}\right)_v Z^{ A 1_{v}}$ and $ 
			\mp \left(\frac{X \pm Z}{\sqrt{2}}\right)_v
			X_{u}
			Z^{ A 1_{u} \oplus 1_{v}}
		 $ for each vertex $v$.
\end{itemize}
If the measurement setting is $X_v Z^{\mathbf{A} 1_v}$, for example, we query prover $v$ with $X$ and the provers for each neighbour of $v$ with $Z$.  The remaining provers are queried with $I$, meaning that we set their outcome to 1.  The combined outcome $a$ is found by multiplying together all of their $\pm 1$ responses.  The timing of these queries does not matter for the test, so we query provers according to the timing they would expect if we were running CALCULATE$(L,x)$.

When we run Algorithm~\ref{algorithm:test}, an individual prover will see only one measurement setting in $\{X, Z, \frac{X\pm Z}{\sqrt{2}}\}$.  Some subset of these measurement settings will be the ones that the prover would expect to see as part of CALCULATE$(L,x)$.  When the prover sees these settings it cannot tell whether we are running CALCULATE$(L,x)$ or TEST$(L,x)$ and so it can only choose some fixed behaviour for each of these measurement settings.  This fact allows us to draw conclusions about CALCULATE$(L,x)$ based on the results of TEST$(L,x)$.  On the other measurement settings the prover can tell that the test is being run, but this will not help it bias the outcome of the calculation.

\begin{algorithm}
\caption{TEST$(L, x)$}
\label{algorithm:test}
\begin{algorithmic}
	\State Pick a measurement settings $M$ uniformly at random from the set $\mathcal{M}(L, x)$
	\State Query provers with the measurement settings $M$ to obtain outcome $a$
	\If {$a = 1$}
		\State \Return ACCEPT
	\Else
		\State \Return REJECT
	\EndIf
\end{algorithmic}
\end{algorithm}

\begin{lemma}
\label{lemma:testproof}
Honest provers cause TEST$(L, x)$ to accept with probability $c_{test}$, and non-communicating quantum dishonest provers cause TEST$(L, x)$ to accept with probability no more than $s_{test}$ where
\begin{equation}
	c_{test} - s_{test} \geq \frac{1}{\text{poly}(|x|)}
\end{equation}

\end{lemma}

\begin{proof}
The full robust proof is found in \cite{McKague:2013:Interactive-proofs-for-BQP-via-self-tested-graph-states}.  Here we give a sketch and in particular we consider only a very special case, where the provers' overall outcome $a$ for each measurement setting $M \in \mathcal{M}(L, x)$ is distributed exactly the same as for the honest provers.  The overview is:
\begin{itemize}
	\item Show that $X^\prime_v$ and $Z^\prime_v$ anti-commute.  This identifies a \emph{logical qubit} within the Hilbert space held by prover $v$, and shows that $X^\prime_v$ and $Z^\prime_v$ operate as the Pauli operators $X$ and $Z$ on this logical qubit.
	\item Show that the group generated by $S^\prime_v$ for various $v$ stabilizes the state $\ket{\psi^\prime}$, establishing that the state on the $n$ logical qubits is the required graph state $\ket{G}$.
	\item Show that the additional measurement operators $\left(\frac{X \pm Z}{\sqrt{2}}\right)^\prime_v$ actually behave like $\frac{X \pm Z}{\sqrt{2}}$ on the logical qubits.
\end{itemize}
Together, these show that the provers are behaving as honest provers would, and hence the calculation will produce a valid result.

First, let us be more specific about the case we are considering.  We suppose that 
\begin{equation}
\bra{\psi^\prime} M^\prime \ket{\psi^\prime} = \bra{G}M\ket{G}
\end{equation}
for each $M \in \mathcal{M}(L,x)$ so that the provers are indistinguishable from honest provers using these measurement settings. Our goal, then, is to show that they are still indistinguishable when we use any adaptive measurement.

First, since $\bra{G} S_v \ket{G} = 1 = \bra{\psi^\prime} S^\prime_v \ket{\psi^\prime}$ and $||S^\prime_v||_\infty = 1$ we have $S^\prime_v \ket{\psi^\prime} = \ket{\psi^\prime}$.  Similarly $-X^\tau Z^{\mathbf{A} \tau} \ket{\psi^\prime} = \ket{\psi^\prime}$ when $\tau$ is the characteristic vector of a triangle $\{u,v,w\}$ in $T$.  We can multiply operators together and still preserve $\ket{\psi^\prime}$, in particular
\begin{equation}
	-X^{\prime \tau} Z^{\prime \mathbf{A} \tau}
	S^\prime_u S^\prime_v S^\prime_w 
	\ket{\psi^\prime} = 
	-X^{\prime \tau} Z^{\prime \mathbf{A} \tau}
	X^{\prime}_u Z^{\prime \mathbf{A} 1_u}
	X^{\prime}_v Z^{\prime \mathbf{A} 1_v}
	X^{\prime}_w Z^{\prime \mathbf{A} 1_v}
	\ket{\psi^\prime}
		=
	\ket{\psi^\prime}
\end{equation}
Now consider a vertex $x \notin \{u,v,w\}$.  If $x$ is a neighbour of an even number of vertices in $\{u,v,w\}$ then $Z^\prime_x$ does not appear in $Z^{\prime \tau}$ and appears in an even number of $S^\prime_u$, $S^\prime_v$ and $S^\prime_w$.  If $x$ is a neighbour of an odd number of vertices in $\{u,v,w\}$ then $Z^\prime_x$ does appear in $Z^{\prime \tau}$ and appears in an odd number of $S^\prime_u$, $S^\prime_v$ and $S^\prime_w$.  In both cases, $Z^\prime_x$ commutes with all other operators and appears an even number of times, giving an overall factor of $I_x$.  Hence the above becomes
\begin{equation}
	-X^{\prime \tau}
	X^{\prime}_u Z^{\prime}_v Z^\prime_w
	X^{\prime}_v Z^{\prime}_u Z^\prime_w
	X^{\prime}_w Z^{\prime}_u Z^\prime_v
	\ket{\psi^\prime}
		=
	\ket{\psi^\prime}.
\end{equation}
Commuting operators on different subsystems past each other, we obtain
\begin{equation}
	-(X^{\prime}_u 
	X^{\prime}_u 
	Z^\prime_u
	Z^\prime_u)(
	Z^\prime_v
	X^{\prime}_v 
	Z^\prime_v
	X^{\prime}_v )(
	X^{\prime}_w
	Z^{\prime}_w Z^\prime_w
	X^{\prime}_w)
	\ket{\psi^\prime}
		=
	\ket{\psi^\prime}.
\end{equation}
Cancelling pairs of operators we finally find 
$	-Z^\prime_v
	X^{\prime}_v 
	Z^\prime_v
	X^{\prime}_v 
	\ket{\psi^\prime} =
	\ket{\psi^\prime}
$,
or equivalently
\begin{equation}
		-Z^\prime_v
	X^{\prime}_v
	\ket{\psi^\prime} = 
	X^{\prime}_v 
	Z^\prime_v
	\ket{\psi^\prime}
\end{equation}
By varying the order of multiplication and the triangle that we use, we can obtain the above relation for any vertex $v$ in the graph.

Now that we have a pair of anti-commuting operators we can use them to identify a logical qubit within the state space of the prover.  It is easiest to see this by moving the logical qubit out of the prover's state space into a new qubit.  The circuit in figure~\ref{fig:epr_local_unitary_circuit} shows how this is done.  The operators $X^\prime_v$ and $Z^\prime_v$ have eigenvalue $\pm 1$ and so are unitaries as are their controlled versions.  The controlled-$Z^\prime_v$ conjugated by Hadamard gates is equivalently a controlled-$X$ gate, targeted on the ancilla qubit.  Hence the circuit is analogous to the familiar SWAP gate and swaps the logical qubit into the new ancilla qubit.

\begin{figure}[h]
\[
\Qcircuit @C=0.5cm @R=0.5cm{
	 \lstick{\text{ancilla}} & \ctrl{1} & \gate{H} & \ctrl{1}  & \gate{H} & \ctrl{1} & \rstick{\text{logical qubit}} \qw  \\
	\lstick{\text{input}} & \gate{X^{\prime}_{v}} & \qw & \gate{Z^{\prime}_{v}} & \qw & \gate{X^{\prime}_{v}} & \rstick{\text{junk}} \qw\\
	}
\]
\caption{Swapping out the logical qubit for prover $v$}
\label{fig:epr_local_unitary_circuit}
\end{figure}
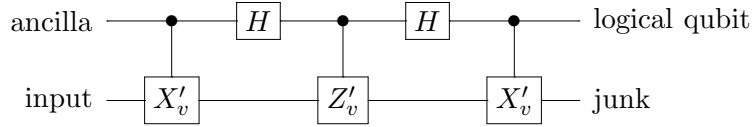

To take this one step further, suppose we apply $X^\prime_v$ before swapping and then move it past the other operators to the right.  The $X^\prime_v$ commutes with the first controlled-$X^\prime_v$ gate, but as we move it past the controlled-$Z^\prime_v$ gate it induces a phase flip on the ancilla qubit.  The two Hadamard gates convert this phase flip to a bit flip, $X$.  At the final controlled-$X^\prime_v$ gate the induced $X$ in turn induces a second $X^\prime_v$, cancelling with the first.  The induced $X$ on the ancilla remains.  Hence applying $X^\prime_v$ before swapping, on the hidden logical qubit, is the same as applying $X$ after swapping, on the swapped out logical qubit now residing in the ancilla.  A similar analysis shows that applying $Z^\prime_v$ before swapping is the same as applying $Z$ after swapping.

To summarize the progress so far, we have a logical qubit buried inside each prover's state space, and we can apply logical $X$ and $Z$ gates to these logical qubits by applying $X^\prime_v$ and $Z^\prime_v$.  Now recall that our provers' state also satisfies $S^\prime_v \ket{\psi^\prime} = \ket{\psi^\prime}$.  This just means that the $n$ provers' state is stabilized by the group generated by $S^\prime_v$.  Breaking up $S^\prime_v$ into $X^\prime_v Z^{\prime \mathbf{A} 1_v}$ and interpreting these as operations on the logical qubits, we see that in fact the state on the $n$ logical qubits is stabilized by the same group as our desired graph state.  Hence the logical qubits must be in that graph state.

Now that we have characterized the state and $X$ and $Z$ measurements, we next consider $\left(\frac{X \pm Z}{\sqrt{2}}\right)^\prime_v$.  First, the honest provers would satisfy $\bra{G}\frac{X_v \pm Z_v}{\sqrt{2}} Z^{\mathbf{A} 1_v}\ket{G} = \bra{G}\frac{X_v \pm Z_v}{\sqrt{2}} X_v \ket{G} = \sfrac{1}{\sqrt{2}}$.  We can see this by using the fact that $S_v \ket{G} = X_v Z^{\mathbf{A} 1_v} \ket{G}  = \ket{G}$ to replace $Z^{\mathbf{A} 1_v}\ket{G}$ with $X_v \ket{G}$.  Next, since $X$, $Z$ and $\ket{G}$ are all written using real entries $\bra{G} Z_v X_v \ket{G} = (\bra{G} Z_v X_v \ket{G})^* = \bra{G}X_v Z_v  \ket{G}$ and by anti-commutation $\bra{G}X_v Z_v  \ket{G} = -\bra{G} Z_v X_v \ket{G}$ so we must have $\bra{G} Z_v X_v \ket{G} = 0$.  So $\bra{G}\frac{X_v \pm Z_v}{\sqrt{2}} Z^{\mathbf{A} 1_v}\ket{G}$ becomes $\sfrac{1}{\sqrt{2}} \bra{G} (X^2_v + Z_v X_v) \ket{G} = \sfrac{1}{\sqrt{2}}$.  Similarly, the honest provers' measurements also satisfy
$ 
	\mp \bra{G}\frac{X_v \pm Z_v}{\sqrt{2}}
			X_{u}
			Z^{ A 1_{u} \oplus 1_{v}}\ket{G} = \sfrac{1}{\sqrt{2}}
$ by using $S_u$ in a similar substitution.

With the assumption that the provers' statistics for these tests match those of the honest provers, we use the analogous substitutions to above to find
$
\bra{\psi^\prime}
\left(\frac{X \pm Z}{\sqrt{2}}\right)^\prime_v X^\prime_v
\ket{\psi^\prime} = \sfrac{1}{\sqrt{2}}
$
and
$
\bra{\psi^\prime}
\left(\frac{X \pm Z}{\sqrt{2}}\right)^\prime_v Z^\prime_v
\ket{\psi^\prime} = \frac{1}{\sqrt{2}}
$
and since $X^\prime_v$ and $Z^\prime_v$ behave like $X$ and $Z$ on the logical qubits, we find 
$\bra{\psi^\prime} Z^\prime_v X^\prime_v \ket{\psi^\prime} = 0$.  Hence the state 
$\left(\frac{X \pm Z}{\sqrt{2}}\right)^\prime_v\ket{\psi^\prime}$
has overlap $\sfrac{1}{\sqrt{2}}$ with each of the two orthogonal states 
$Z^\prime_v
\ket{\psi^\prime}$ and 
$X^\prime_v \ket{\psi^\prime}$ meaning that 
\begin{equation}
	\left(\frac{X \pm Z}{\sqrt{2}}\right)^\prime_v\ket{\psi^\prime} = 
	\frac{X^\prime_v \pm Z^\prime_v}{\sqrt{2}} \ket{\psi^\prime}
\end{equation}
i.e.\ $\left(\frac{X \pm Z}{\sqrt{2}}\right)^\prime_v$ operates on the logical qubit as $\frac{X \pm Z}{\sqrt{2}}$.

Now we can see that the provers are essentially implementing the honest protocol.  Somewhere in their state there are logical qubits in the required graph state, and their measurements on this state are just $X$, $Z$ and $\frac{X \pm Z}{\sqrt{2}}$.  Hence any calculation we do with these resources will give the same results as for the honest provers.  This proves (for this limited case) that if the provers cause TEST$(L, x)$ to pass with the same probability as honest provers, then CALCULATE$(L, x)$ will produce a correct result.

\end{proof}


\section{Discussion}

Our construction has a somewhat remarkable property.  Since triangular cluster states are universal, the state preparation depends only on the size of the calculation.  The quantum provers are also constant, requiring only the ability to measure in some fixed set of bases.  This means that TEST$(L,x)$ also only depends on the size of the calculation, since it depends only on the state and the measurements.  The classical verifier is also rather simple.  It can be given a circuit as its input from which it reads off what gates to perform and simply looks up what angle to measure for that gate.  The remaining calculations are simply XORs.  This is a clear example of where the simplicity of the measurement-based quantum computing model allows for a simple analysis.

Another interesting property is that no single quantum prover has enough power to convince even \emph{itself} whether the input is in the language.  Indeed, all the quantum parts together, including the state preparation, still cannot perform even simple calculations since they can only prepare and measure some fixed state.  They lack the capacity to perform the XORs required to perform a full measurement-based calculation.  It is only when we combine the verifier, provers, and state preparation together that we obtain enough power to perform any substantial calculations.

Our construction, as with all other proposals to date, requires that the provers are quantum, meaning that we assume quantum mechanics applies to the provers.  If we are to use interactive proofs in experimentally verifying quantum mechanics then this may lead to a circular argument if we are not careful.  If we, assuming quantum mechanics is correct, perform an interactive proof with a quantum computer to make predictions about a large quantum experiment, and the experiment agrees with the predictions, is this evidence that quantum mechanics is correct?  Conceivably, both the computation and the experiment disagree with quantum mechanics but still agree with each other.  A much more satisfying result would be to have an interactive proof where we do not assume that the provers are quantum, but where there exist efficient quantum provers.  We can still use our results, however: we can view the interactive proof with the quantum provers as being \emph{itself} an experiment, i.e.\ we can form a prediction like ``if quantum mechanics is correct, then our interaction with the experiment should satisfy these properties...''  Significant deviation from the expected behaviour then signals a potential problem with the theory.

Currently our construction uses many constant sized provers, providing a nice complement to Reichardt et al.'s result using a constant number of provers.  Much of our result could easily be adapted to the case of two provers.  The most difficult part is the graph-state test.  Likely it is not possible to prove a self-testing theorem for two provers if there are any odd cycles in the graph since it would be necessary at some point to test the entanglement across an edge with both vertices held by a single prover.  However, bipartite graph states could yet be self-tested with two provers.  

\vspace{0.2cm}
\noindent
\textbf{Acknowledgements - }
Thanks to Serge Massar, Stefano Pironio, Iordanis Kerenidis, Fr\'{e}d\'{e}ric Magniez, David Hutchinson and Michael Albert for helpful discussions.  This work is funded by the Centre for Quantum Technologies, which is funded by the Singapore Ministry of Education and the Singapore National Research Foundation, and by the University of Otago.

\bibliographystyle{halphads}
\bibliography{Global_Bibliography}

\end{document}